\documentclass[12pt,reqno]{amsart}
\usepackage{amsaddr}
\usepackage{hyperref}
\usepackage{graphicx}
\usepackage{setspace}
\usepackage{empheq}
\usepackage{cite}
\usepackage{textcomp}
\usepackage{amssymb}

\usepackage{color}
\definecolor{MyLinkColor}{rgb}{0,0,0.4}

\newcommand{\dx}{\mathrm{d}}

\newtheorem{thm}{Theorem}[section]

\theoremstyle{remark} 
\newtheorem{rem}[thm]{Remark}

\numberwithin{equation}{section}   

\title[A steady stratified purely azimuthal flow representing the ACC]{A steady stratified purely azimuthal flow representing the Antarctic Circumpolar Current}

\author[C.~I.~Martin]{Calin Iulian Martin}
\address{Faculty of Mathematics, University of Vienna, Austria \\
	E-mail address: \href{mailto:calin.martin@univie.ac.at}{\textsf{calin.martin@univie.ac.at}}}

\author[R.~Quirchmayr]{Ronald Quirchmayr$^*$}
\address{Department of Mathematics, KTH Royal Institute of Technology, Sweden \\
	E-mail address: \href{mailto:ronaldq@kth.se}{\textsf{ronaldq@kth.se}}}

\thanks{$^*$Corresponding author}

\begin{document}
\maketitle

\begin{abstract}
We construct an explicit steady stratified purely azimuthal flow for the governing equations of geophysical fluid dynamics. These equations are considered in a setting that applies to the Antarctic Circumpolar Current, 
accounting for eddy viscosity and forcing terms.

\vspace{1em}
\end{abstract}

\section{Introduction}
\noindent
Thorough analytical investigations of exact solutions to the fully nonlinear governing equations of geophysical fluid dynamics (GFD)  represent  an extensive and active research area, which was initiated by Constantin \cite{CGeoResLett, CoGeoPhys, CoPhysOc13, CoPhysOc14} and Constantin and Johnson \cite{CJ, CJaz, CJazAcc, CJPoF}.
Following this approach, we construct a unidirectional flow satisfying GFD considered in the so-called $f$-plane at the 45th parallel south, enhanced with an eddy viscosity term and a forcing term, and equipped with appropriate boundary conditions. 
We propose this specific flow for representing the gross dynamics of the Antarctic Circumpolar Current (ACC)---the World's longest and strongest Ocean current. 

ACC has no continental barriers: it encircles Antarctica along a 23 000 km path around the polar axis towards East at latitudes between 40° to 60°, see Fig.~\ref{fig: fronts}. It thereby links the Atlantic, Pacific and Indian Oceans making it the most important oceanic current in the Earth’s climate system.
ACC has a rich and complicated structure. Many factors contribute to its complex behavior---the most important driver are the strong westerly winds in the Southern Ocean region. In addition to that there exist mesoscale eddies of a size up to 100 km, which transport the wind-induced surface stress to the bottom and also enable meridional mass transport; there are sharp changes in water density due to variations in temperature and salinity---known as fronts or jets---located at ACC's boundaries (see Fig.~\ref{fig: fronts}); ACC is strongly constrained by the bottom topography; there are observed variations in time such as the Antarctic Circumpolar Wave; etc. We refer to \cite{Danabas, Farneti_etal2010, FCM04, IR96, KN01, Marsh16,TG94, Thompson08, WP96, W99} for further information about the geophysical aspects and modeling as well as observational data and simulations for ACC. 

From an analytical perspective one is forced to largely, yet reasonably, simplify the geophysical scenario to obtain a tractable model, which---in the ideal case---exhibits exact and explicit solutions opening the path for an in-depth analysis. 
Thus we do not account for all of the before mentioned phenomena, but assume a steady flow in purely azimuthal direction, which is vertically bounded by a flat bottom and a flat ocean surface. By considering Euler's equation of motion in the $f$-plane, we obtain a valid approximation of the Coriolis effects close to the 45th parallel south; in this way the Earth's curvature is neglected and no boundaries in the meridional direction are assumed. To account for the transportation effects of mesoscale eddies we equip the system with an eddy viscosity term; furthermore we include a forcing term to ensure the dynamical balance of the flow.
Both pressure and wind stress are prescribed on the ocean surface; a no-slip boundary condition is assumed for the ocean bed.

A similar setting has recently been considered in \cite{Quir}, where an explicit solution in terms of a given viscosity function was presented. In this note at hand we extend these results to stratified flows, i.e.~we do account for variations of the water density (with depth and latitude). The established explicit solution is an analytic function of both the viscosity function and the density distribution.
A collection of numerous recent related analytical studies can be found in \cite{CJazAcc, HazMary, HazDCDS, HsuMarAcc, Mary, Quir} and the references therein.

\begin{figure}[h]
\begin{center}
\includegraphics*[width=0.65\textwidth]{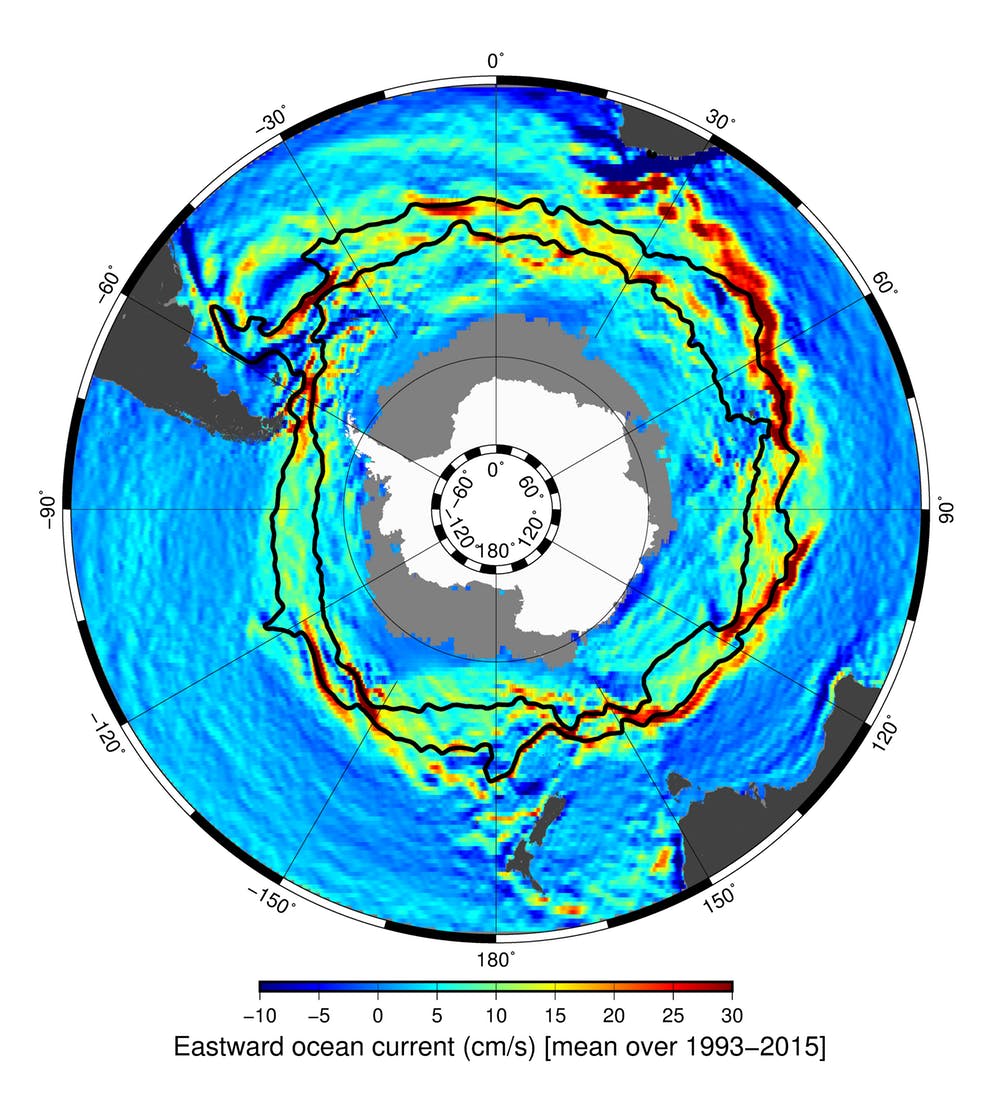}
\end{center}
\caption{Variations in density and the mean eastward current speed. The black lines in the image represent the fronts of ACC. Image credit: Hellen Phillips (Senior Research Fellow, Institute for Marine and Antarctic Studies, 
University of Tasmania), Benoit Legresy (CSIRO) and Nathan Bindoff (Professor of Physical Oceanography, Institute for Marine and Antarctic Studies, University of Tasmania)} 
\label{fig: fronts}
\end{figure}

\section{Model under study}
\noindent
We introduce the governing equations for geophysical ocean flows that set the basis for our study:
we take into account the effects of the Earth's rotation on the fluid body by choosing a rotating framework with 
the origin at a point on the Earth's surface. Accordingly, we will work with Cartesian coordinates $(x,y,z)$, where $x$ denotes the direction of increasing longitude, $y$ is the direction of increasing latitude and $z$ represents the local vertical,
respectively. Denoting with $t$ the time variable, and with $(u(x,y,z,t),v(x,y,z,t),w(x,y,z,t))$ the velocity field, the governing equations 
for inviscid and incompressible geophysical ocean flows at latitude $\phi$ are (cf.~\cite{CoPhysOc13, Pedl, Val}) 
 the Euler equations 
\begin{align}\label{full_geophys_eq}
\begin{aligned}
  u_{t}+u u_{x}+v u_{y}+w u_{z} +2\Omega (w\cos\phi- v\sin\phi) & =  - \frac{1}{\rho} P_x,\\
  v_{t}+u v_{x}+v v_{y}+w v_{z} +2\Omega u\sin\phi & =  - \frac{1}{\rho} P_y,\\
 w_{t}+u w_{x}+v w_{y}+ w w_{z} -2\Omega u\cos\phi & = - \frac{1}{\rho} P_z-g,
  \end{aligned}
\end{align}
and  the equation of mass conservation
\begin{equation}\label{cont_eq}
 {\rm div}(\rho {\bf u})=0,
\end{equation}
where ${\bf u}=(u,v,w)$. Here, $P=P(x,y,z,t)$ denotes the pressure field, $\rho=\rho(y,z)$ serves a (prescribed) density distribution to account for density gradients, mainly caused by variations of salinity and temperature, $\Omega \approx 7.29\cdot 10^{-5} \, \mathrm{rad} \, s^{-1}$ is the Earth's (constant) rotational speed around the polar axis toward the East, and $g\approx 9.82 \, m/s^2$ denotes the gravitational constant.
The governing equations \eqref{full_geophys_eq}--\eqref{cont_eq} hold throughout the fluid domain, which lies between the rigid flat bed at $z=-d$ ($d$ is the constant water depth) and above the flat surface situated at $z=0$. 

Our aim is to derive exact formulas for purely azimuthal flows (i.e.~$v=w=0$) in the region of the ACC (i.e.~we take $\phi=  -\pi/4$, which corresponds to the 45th parallel south). 
Furthermore we incorporate the transfer of the wind-generated surface stress to the bottom, which is due to the presence of mesoscale eddies,
by adding a viscosity term of the form $(\nu u_z)_z$ to the right hand side of the first equation in \eqref{full_geophys_eq}; the coefficient $\nu=\nu(z)$ is a smooth function of depth being strictly greater than some positive constant, see~\cite{IR96}.
The classical model of uniform eddy viscosity is due to~\cite{Stommel1960}. We follow the more realistic approach with a depth dependent viscosity function as it was introduced in~\cite{CK09}.
Finally, we include a forcing term $F=F(y,z)$ to guarantee non-trivial solutions, cf.~\cite{CJazAcc, How, Quir}. 

According to the previous considerations, we will consider the following set of equations governing the ACC as a geophysical flow model in the purely azimuthal direction:
\begin{align}
0&=-\frac{1}{\rho}P_x+(\nu u_z)_z \label{one} \\
-\sqrt{2}\Omega u &=-\frac{1}{\rho}P_y +F \label{two} \\
-\sqrt{2}\Omega u & =-\frac{1}{\rho}P_z-g \label{three} \\
u_x &=0\label{four}
\end{align}
and are valid within the fluid domain $\mathcal{D}:=\mathbb{R}^2\times [-d,0]\subset\mathbb{R}^3$.

The equations of motion \eqref{one}-\eqref{four} are supplemented by the following boundary conditions:
\begin{align}
	u=0 \quad  &\text{on} \quad z=-d,\label{bottom_cond}\\
	P=P_{\rm atm} \quad  &\text{on} \quad  z=0,\label{surf_cond}\\
	\tau=\tau_0 \quad  &\text{on} \quad z=0, \label{wstress}
\end{align}
where $\tau(y,z):=\rho\nu u_z$ represents the \emph{wind stress}; i.e.~we assume a no-slip bottom and constant pressure as well as wind stress at the surface.

\section{Explicit solution}
\noindent
\begin{thm} \label{prop}
	The solution $(u,P)$ of system \eqref{one}--\eqref{four} with  boundary conditions \eqref{bottom_cond}--\eqref{wstress} is given by 
	\begin{equation}\label{formula_u}
	u(y,z)=\frac{\tau_0}{\rho(y,0)}\int_{-d}^z\frac{\dx \zeta}{\nu(\zeta)} 
	\end{equation} 
	and
\begin{align}\label{formula_P}
	\begin{aligned}
		P(y,z)=  \sqrt{2}\Omega \int_{0}^z \rho(y,\zeta) u(y,\zeta) \, \dx \zeta -  g \int_{0}^z  \rho(y,\zeta) \, \dx \zeta +
		P_{\rm atm},
	\end{aligned}
\end{align} 
	for $(y,z)\in\mathbb{R}\times [-d,0]$. Furthermore, the forcing term $F$ can be recovered from \eqref{two} by means 
	of \eqref{formula_u} and \eqref{formula_P}.
\end{thm}

\begin{proof}
Utilizing \eqref{four} we differentiate by $x$ in \eqref{one}-\eqref{three} and obtain that
\begin{equation}\label{grad_Px}
\nabla P_x=0\,\,{\rm in}\,\,\mathcal{D}.
\end{equation}
From the condition $P=P_{atm}$ on the surface $z=0$ we obtain that $P_x=0$ on $z=0$. Thus, from \eqref{grad_Px} we see that 
\begin{equation*}
 P_x=0\,\,{\rm in}\,\,\mathcal{D}.
\end{equation*}
Hence, equation \eqref{one} becomes $(\nu u_z)_z=0$, which implies (making also use of \eqref{four}) that there is a funtion $y\mapsto B(y)$ such that 
\begin{equation*}
\nu u_z=B(y)\,\, \text{in}\,\,\mathcal{D}.
\end{equation*}
The latter equation and condition \eqref{wstress} imply that $$\nu(0)u_z(y,0)=\frac{\tau_0}{\rho(y,0)}=B(y)\,\,{\rm for}\,\,{\rm all}\,\,y.$$ 
Therefore, $\nu(z)u_z(y,z)=\frac{\tau_0}{\rho(y,0)}$ for all $y$ and for all $z\in [-d,0]$. Thus, by means of the bottom boundary condition \eqref{bottom_cond} we infer that $u$ satisfies \eqref{formula_u}.

Integrating with respect to $z$ in \eqref{three} we obtain that
\begin{equation*}
 P(y,z)=  \sqrt{2}\Omega \int_{-d}^z \rho(y,\zeta) u(y,\zeta) \, \dx \zeta -  g \int_{-d}^z  \rho(y,\zeta) \, \dx \zeta + C(y),
\end{equation*}
where the integration constant $C(y)$ is determined by \eqref{surf_cond} and satisfies 
\begin{equation*}
 C(y)=P_{\rm atm} -\sqrt{2}\Omega \int_{-d}^0 \rho(y,z) u(y,z) \, \dx z +  g \int_{-d}^0  \rho(y,z) \, \dx z.
\end{equation*}
Therefore, $P$ satisfies \eqref{formula_P} for all $(y,z)\in\mathbb{R}\times [-d,0]$.
\end{proof}

\begin{rem}
One immediate consequence of \eqref{formula_u} is that the vorticity vector associated with the flow \eqref{formula_u} is $(0,u_z, -u_y)$ has a non-vanishing second and third component. 
This represents a marked difference, if compared with the case of homogeneous flows (considered in \cite{Quir}) where only the middle component $u_z$ survives, the first and the third being zero because of the lack of $y$ dependence of $\rho$. Thus, allowing for significant variations in density leads to solutions that exhibit substantial shear not only in the vertical direction but also in the latitudinal direction, as well.
\end{rem}

\vspace{1em}
\noindent
\textbf{Acknowledgments}\\
\noindent
C.~I.~Martin would like to acknowledge the support of the Austrian Science Fund (FWF) under research grant P 30878-N32. R.~Quirchmayr acknowledges the support of FWF under research grant J 4339-N32.

\end{document}